\documentclass[twoside,leqno,11pt]{article}
\usepackage{graphicx,fullpage}
\usepackage{amsfonts,amsmath,amsthm,amssymb,mathtools, comment}\usepackage{xcolor}
\usepackage[T1]{fontenc}
\usepackage[colorlinks,citecolor=blue,linkcolor=blue,urlcolor=red,pagebackref]{hyperref}
\usepackage[capitalise]{cleveref}
\usepackage{url}
\usepackage{thmtools, thm-restate}
\usepackage{algorithm, algorithmicx, algpseudocode}
\usepackage{tikz}
\usepackage{bm}
\usepackage{thm-restate}
\usepackage{mdframed}
\usepackage{multirow}

\newcommand{\eps}{\varepsilon}

\newcommand{\poly}{\operatorname{\mathrm{poly}}}

\newcommand{\ceil}[1]{\left\lceil #1 \right\rceil}
\newcommand{\Oh}{\mathcal{O}}
\newcommand{\Ohtilde}{\tilde{\mathcal{O}}}

\newcommand{\mindiam}{\texttt{min-diam}}
\newcommand{\minradius}{\texttt{min-radius}}

\newcommand{\dmin}{d_{\min}}
\newcommand{\mindiamalg}{\operatorname*{\mathbf{MinDiameter}}}
\newcommand{\singletypealg}{\operatorname*{\mathbf{SingleType}}}

\newtheorem{theorem}{Theorem}[section]
\newtheorem{lemma}[theorem]{Lemma}

\newtheorem{claim}[theorem]{Claim}
\newtheorem{cor}[theorem]{Corollary}
\theoremstyle{definition}
\newtheorem{definition}[theorem]{Definition}

\title{Improved Approximation of Min-Distances in Near-Linear Time}
\author{Yael Kirkpatrick\thanks{\texttt{yaelkirk@mit.edu}, supported by NSF Grant No 2141064.}\\MIT}
\date{}

\begin{document}
    \maketitle
    \begin{abstract}
        We study the problem of approximating the diameter of directed graphs under the min-distance measure, defined as $d_{\min}(u,v) = \min(d(u,v), d(v,u))$. Unlike standard shortest-path distance, min-distance is not a metric, which renders many classical techniques inapplicable. Prior work has therefore focused on approximating this parameter, culminating in an approximation-runtime tradeoff  by  Dalirrooyfard et al. [ICALP'19] giving a $(4k-1)$-approximation in $\tilde{O}(mn^{1/(k+1)})$ time for any positive integer $k$ and, more recently, the first near-linear time constant approximation by Chechik and Zhang [FOCS'22], where they obtained a 4-approximation to the min-diameter.

In this work we present a randomized near-linear time algorithm that achieves a $3$-approximation to the min-diameter, outperforming all known approximation-runtime tradeoffs. Our approach introduces a novel type-classification framework that may be of independent interest.

    \end{abstract}
    \section{Introduction}
The graph \emph{diameter}, defined as the maximum shortest-path distance between pairs of vertices, is one of the most central and widely studied graph parameters. The diameter serves as an indicator of graph complexity, and consequently many practical applications involve estimating it \cite{realworlddiam12, empiricalboundsdiam09, smallnetworkdsdiam11, bfsrealgraphsdiameter15}. As computing the exact diameter is hard under fine-grained assumptions \cite{sprasediamradius13}, much work has gone into constructing approximation algorithms \cite{aingworthdiam99, sprasediamradius13, towardsapprox21, betterapproxdiam14, cgr}, studying the hardness of approximating this value \cite{undirdiamlb, towardsapprox21, lidiamapprox21, bonnet21, dirdiamlb21} as well as studying it in specialized settings \cite{chepoi1994linear, realweightedpath05, distributeddiam12, dynamicdiam19}.

In undirected graphs, the natural distance between a pair of points $u,v$ is the length of the shortest path between them, $d(u,v)$. In directed graphs, however, this notion is no longer symmetric and the shortest path (or one-way) distance, from $u$ to $v$, $d(u,v)$, may differ significantly from the distance in the opposite direction, $d(v,u)$. In the context of the graph diameter, this asymmetry can often lead to unsatisfying answers, as the diameter of a graph with more than one strongly connected component is infinite under this one-way notion of distance. This in part has motivated the study of various notions of distance in directed graphs, including the roundtrip distance, $d(u,v) + d(v,u)$ \cite{roundtrip99}, the max-distance $d_{\max}(u,v)\coloneqq \max(d(u,v), d(v,u))$ \cite{fixedparam16} and the min-distance $\dmin(u,v)\coloneqq \min(d(u,v), d(v,u))$ \cite{fixedparam16}.

As with the standard definition of distance, computing the diameter under these various notions of distance cannot be computed in $O(m^{2-\eps})$ time for any $\eps > 0$ under the Strong Exponential Time Hypothesis (SETH) \cite{fixedparam16, sprasediamradius13} so we turn to approximation algorithms instead. Among these notions, the roundtrip distance and the max-distance are metrics, and therefore many algorithms for the traditional diameter extend naturally to these settings as well \cite{towardsapprox21, betterapproxdiam14, sprasediamradius13}. The min-distance, in contrast, is not a metric since it does not satisfy the triangle inequality. This has inspired a long line of work on approximating graph parameters in the min-distance setting, as no existing algorithms extend to this setting directly.

In this paper we study the min-distance measure, focusing in particular on the diameter of a graph under this notion, known as the min-diameter. This distance measure naturally captures real-world scenarios, such as determining the fastest way for a patient to receive care - either by traveling to a hospital or by having a doctor come to them.

In the paper that introduced the concept of min-distance, Abboud, Vassilevska W. and Wang \cite{fixedparam16} showed a simple 2-approximation to the min-diameter in near-linear\footnote{By near-linear time we mean $O(m\poly \log n)$.} time in directed acyclic graphs (DAGs). Follow-up work has obtained further results for the DAG setting  \cite{mindiamdag21, bergermindiam2023}.

For general graphs, a linear-time $2$-approximation appears to be far beyond reach. The best known lower bound for the problem was shown by Berger, Kaufmann and Vassilevska Williams \cite{bergermindiam2023}, where they showed any $(2-\varepsilon)$-approximation for min-diameter requires $\Omega(m^{2-o(1)})$ time. Initially, it was unclear whether any nontrivial approximation could be achieved in subquadratic, let alone near-linear, time. In 2019, Dalirrooyfard et al.~\cite{mindistance2019} gave the first approximation algorithm for the min-diameter in general graphs, establishing a tradeoff between approximation quality and runtime: for any integer $k \geq 1$, they obtained a $(4k-1)$-approximation in $\Ohtilde(mn^{1/(k+1)})$ time\footnote{We use $\Ohtilde$ to hide polylogarithmic factors in the runtime.}. In particular, their result yields an $O(\log n)$-approximation in near-linear time, while the best approximation factor they achieve is $3$ in $\tilde{O}(m\sqrt{n})$ time. It remained open whether a constant-factor approximation could be achieved in near-linear time. This question was resolved by Chechik and Zhang~\cite{mindiam2022}, who obtained a $4$-approximation to the min-diameter in near-linear time.

\subsection{Our Results}
Our results in the context of current best-known algorithms are summarized in \cref{tab:results}. Our main result is an improved approximation to min-diameter, obtaining a near-linear time 3-approximation.

\begin{table}[]
\center
\renewcommand{\arraystretch}{1.3}
\begin{tabular}{|l|c|c|c|}
\hline
\textbf{Problem}                          & \textbf{Runtime}     & \multicolumn{1}{l|}{\textbf{Approximation}} & \textbf{Reference}        \\ \hline
\multirow{3}{*}{min-diameter}             & $\Ohtilde(mn^{1/(k+1)})$ & $4k-1$                              & \cite{mindistance2019} \\ \cline{2-4}
                                          & $\Ohtilde(m)$        & 4                                           & \cite{mindiam2022}       \\ \cline{2-4}
                                          & $\Ohtilde(m)$        & 3                                           & This paper                \\ \hline
\end{tabular}
\caption{Summary of Results.}\label{tab:results}
\end{table}

\begin{theorem}\label{thm:mindiam}
    There is a randomized algorithm that computes a 3-approximation to the min-diameter in $\Ohtilde(m)$ time with high probability.
\end{theorem}

This result strictly improves upon the entire tradeoff of Dalirrooyfard et al.~\cite{mindistance2019} for all values of $k$, achieving an equal or better approximation ratio with faster runtime. In addition, it improves upon the $4$-approximation of Chechik and Zhang~\cite{mindiam2022} by reducing the approximation factor to $3$.

To obtain this result we develop a type-classification framework which could be of independent interest to the study of algorithmic questions in the min-distance setting or other non-metric distance measures.

\subsection{Technical Overview}

Given a threshold $D$, our goal is to either find a pair of vertices $(u,v)$ such that $\dmin(u,v) \geq D/3$, or determine that $\mindiam(G)<D$.  By performing a binary search over $D$, we can identify the optimal threshold and thereby obtain a $3$-approximation to the min-diameter.

Starting from the full set of vertices, a natural approach, used in many previous min-diameter approximation algorithms, is to partition the set into two balanced subsets and recurse (e.g. \cite{fixedparam16}). To this end, we sample a vertex $x$ and compute its incoming and outgoing balls of radius $D/3$, namely $C_1 = B^-(x, D/3)$ and $C_2 = B^+(x, D/3)$. If it is the case that $\mindiam(G) \geq D$, then there exists a pair of diameter endpoints $(s,t)$ such that $\dmin(s,t) \geq D$. This pair must lie entirely within either $C_1$ or $C_2$, since otherwise there would exist a path of length less than $2D/3$ between them, contradicting the assumption. Consequently, by recursing on $C_1$ and $C_2$, we are guaranteed that at least one of the two sets contains a pair of vertices at min-distance at least $D$.

This approach appears promising. However, recursing on these sets does not guarantee correctness: if we later find a pair of vertices with min-distance at least $D/3$, it is unclear whether their distance was already large in the original graph or became large due to the removal of some of the vertices along the shortest path between them. To address this issue, we adopt the idea of Chechik and Zhang~\cite{mindiam2022} and introduce \emph{padding} vertices in the recursive calls - vertices that are not considered as potential diameter endpoints but may lie on shortest paths between such pairs.

In our setting, we include $P_1 \coloneqq B^-(C_1, D/3)$ as padding for $C_1$ and $P_2 \coloneqq B^+(C_2, D/3)$ as padding for $C_2$. We then recursively invoke the algorithm on the pairs $(C_1, P_1)$ and $(C_2, P_2)$.
Computing single-source shortest paths (SSSP) from vertices in $C_1$ using only edges induced by $P_1$ ensures that any path of length less than $D/3$ between two vertices in $C_1$ is preserved. Consequently, if we find a pair $u,v \in C_1$ such that $\dmin^{G[P_1]}(u,v) \geq D/3$, then it follows that $\dmin^{G}(u,v) \geq D/3$. Otherwise, either $d_G(u,v) < D/3$ or $d_G(v,u) < D/3$, and in both cases the corresponding shortest path is fully contained in $P_1$.

Having addressed correctness, we now face a different issue: the recursive sets may overlap, potentially causing a rapid blow-up in the runtime. To address this problem, we introduce a parameter $L \approx \log n$ and only recurse on $(C_1, P_1)$ and $(C_2, P_2)$ when the overlap between $P_1$ and $P_2$ involves fewer than $m/L$ edges.
We choose $L$ to upper bound the recursion depth, allowing for a multiplicative $(1 + 1/L)$ increase in the number of edges at each level. Over at most $L$ levels of recursion, this results in a total blow-up of at most $(1 + 1/L)^L = \Oh(1)$ in the number of edges.
If the overlap is large, we instead handle the vertices associated with the overlap separately (see below). Since this case arises only when the overlap is large, we can argue that after at most $L$ attempts, we will find a vertex $x$ that yields the desired partition with small overlap.

To handle the large overlap we introduce a new concept - classifying vertices into one of two types. We say that a vertex $v$ is \emph{degenerate} if its min-distance to either diameter endpoint is at least $D/3$. In this case, running SSSP from $v$ would already find a pair of points of the desired distance.
For a non-degenerate vertex $v$, it must hold that either both $d(s,v) < D/3$ and $d(t,v) < D/3$, or both $d(v,s) < D/3$ and $d(v,t) < D/3$. We refer to vertices satisfying the former condition as \emph{Type 1}, denoting their set by $T_1$, and those satisfying the latter as \emph{Type 2}, denoting their set by $T_2$.

Now consider two non-degenerate vertices $u$ and $v$ such that $d(u,v) < D/3$. If $u \in T_1$, then necessarily $v \in T_1$, as otherwise the path $s \rightsquigarrow u \rightsquigarrow v \rightsquigarrow t$ would have length less than $D$, contradicting the choice of $s,t$; see \cref{fig:types}. Similarly, if $v \in T_2$, then $u \in T_2$.
Therefore, if there exists a complete cycle of non-degenerate vertices $u_1, u_2, \ldots, u_k$ such that
\[
d(u_1,u_2), \ldots, d(u_{k-1},u_k), d(u_k,u_1) < D/3,
\]
then all vertices $u_1, \ldots, u_k$ must belong to the same type.

\begin{figure}[ht]
    \centering
    \includegraphics[width=0.3\textwidth]{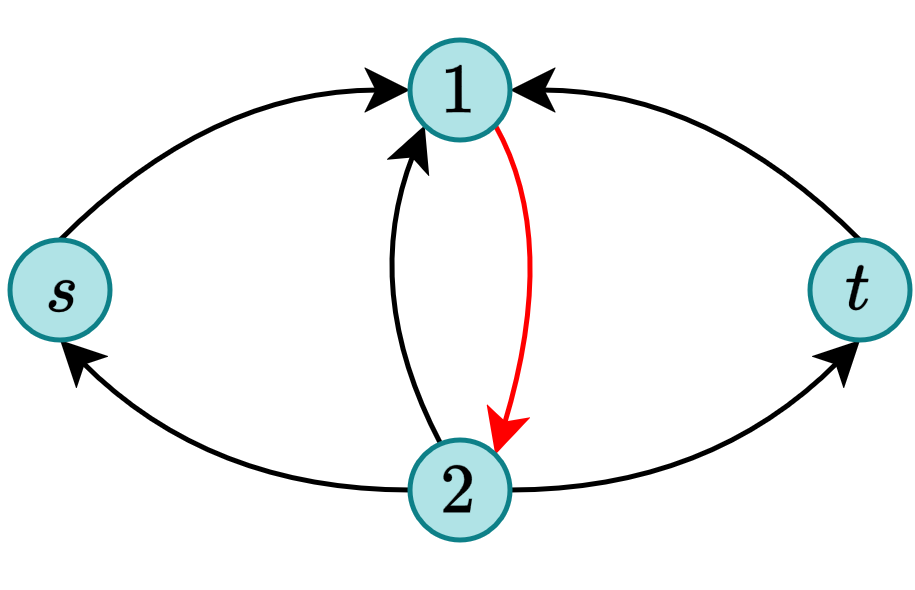}
    \caption{\small Distances between $s,t$ and vertices of Type 1 and Type 2. Black edges denote paths of length $<D/3$, while the red edge denotes that a path of length $<D/3$ cannot exist.}
    \label{fig:types}
\end{figure}

Thus, suppose we obtain a partition $(C_1, P_1), (C_2, P_2)$ with a large overlap $W_x \coloneqq P_1 \cap P_2$, where $|E[W_x]| > m/L$. Intuitively, this overlap corresponds to a large set of vertices that all share the same type.
If all vertices in $W_x$ are of Type~1, we repeatedly sample $\Theta(\log n)$ vertices and take the intersection of their incoming $D/3$ neighborhoods, retaining $s$ or $t$ if they were already present. A lemma from prior works~\cite{directedgirth2020, chechik2020girth} ensures that, with high probability, each iteration shrinks the set by a constant factor. Consequently, after $\Theta(\log n)$ iterations, the set becomes sufficiently small, allowing us to identify $s$ or $t$.
An analogous procedure applies if all vertices are of Type~2, using intersections of outgoing $D/3$ neighborhoods. If no pair of vertices with min-distance at least $D/3$ is found, we conclude that neither $s$ nor $t$ lies in the overlap, and we remove the corresponding vertices from further consideration.
More precisely, we consider the sets $C_1 \cap B^+(W_x, 2D/3)$ and $C_2 \cap B^-(W_x, 2D/3)$, whose vertices are guaranteed to share the same type as $x$, and remove them from $C$ if neither $s$ nor $t$ is found.

We now repeat the process by sampling a new vertex $x' \in C$ and constructing the corresponding sets $C'_1, C'_2, P'_1, P'_2$. It is again possible that the resulting overlap $W_{x'}$ is large. However, we observe that $W_{x'}$ must be disjoint from $W_x$, since otherwise $x'$ would have been removed from $C$ in the previous iteration.
This disjointness property ensures that after at most $L$ such attempts, we obtain a vertex $x'$ whose corresponding overlap is small, allowing us to recurse while incurring only a $(1+1/L)$ multiplicative increase in the total number of edges.

The approach outlined above has one remaining issue. While initially we may assume that any vertex from which we run SSSP is non-degenerate, since otherwise we would already be done, this assumption no longer holds after recursion. In particular, the recursion partitions vertices into a set of potential diameter endpoints $C$ and padding vertices $P$. While we are guaranteed that any pair of vertices in $C$ with min-distance at least $D/3$ is also far apart in the original graph, this guarantee does not extend to pairs of vertices in $P$. Consequently, vertices in the overlap $W_x$ may be degenerate, since $W_x$ may consist entirely of padding vertices.

To address this issue, we propagate multiple layers of padding with each recursive call, namely $P_1, \ldots, P_{L+1}$. After the first call, we can ensure that any vertex explored in $C, P_1, \ldots, P_L$ is non-degenerate (otherwise we would have already terminated), since all paths of length less than $D/3$ between such vertices are preserved in the final padding layer $P_{L+1}$.

Applying the previous arguments, we can therefore identify a vertex $x$ such that the $L$ layers of padding constructed around it have small overlap. Upon recursing, we maintain that vertices in $C$ together with the first $L-1$ padding layers remain non-degenerate, and this invariant continues inductively. Although we lose one layer of padding at each level of recursion, the choice of $L$ as an upper bound on the recursion depth ensures that, even at the final level, all paths of length less than $D/3$ between vertices in $C$ are preserved.

\subsection{Discussion and Open Problems}
In prior work on graph parameters under the min-distance measure, techniques for approximating the min-diameter have often extended to the min-radius and, more generally, to approximating all min-eccentricities. In particular, the previous best min-diameter algorithms admit matching results for min-radius: a $3$-approximation in $\Ohtilde(m\sqrt{n})$ time~\cite{mindistance2019}, and a $4$-approximation in near-linear time~\cite{mindiam2022}. However, it is unclear whether the techniques developed in this work extend to the radius setting.

The key idea in our algorithm, the type-classification framework, admits an analogous formulation for the radius problem. Suppose $\minradius(G)\leq R$, let $c$ be a center, and consider the task of finding a vertex with min-eccentricity at most $3R$, which we call \emph{degenerate}. Every vertex $u$ satisfies either $d(u,c)\leq R$ or $d(c,u)\leq R$; we refer to vertices satisfying the former as Type~1 and the latter as Type~2. The main observation is that if there exists a path of length at most $R$ from a Type~2 vertex to a Type~1 vertex, then \emph{both} vertices are degenerate. This yields structural properties analogous to those used in our min-diameter approximation algorithm.

This suggests that our techniques may extend to the radius setting. However, the key step in our algorithm, using $L$ layers of padding to ensure that degenerate vertices in the padding yield a valid solution, does not appear to generalize.
Thus, obtaining a near-linear-time $3$-approximation for the min-radius, and more generally improving approximations for min-eccentricities, remains an interesting open problem.

\subsection{Organization}
In \cref{sec:prelim} we introduce the necessary definitions and preliminary lemmas. In \cref{sec:type} we develop a type-classification framework, which we then use in \cref{sec:mindiam} to prove \cref{thm:mindiam}.

    \section{Preliminaries}\label{sec:prelim}
Let $G = (V,E)$ be a graph with $n = |V|$ vertices and $m = |E|$ edges. We will use $m$ throughout the paper to refer to the number of edges in the current instance of the graph we are working with (which might be significantly smaller than the graph we started out with). When referring to the number of edges in some former graph we use $|E|$ to avoid ambiguity.

Denote by $d_G(u,v)$ the length of the shortest path from $u$ to $v$ in $G$. We omit $G$ from the notation when it is clear from context. Given a vertex $u$ and value $r$ denote by $B^+(u,r),B^-(u,r)$ the outgoing and incoming balls of radius $r$ around $u$, i.e., $B^+(u,r) \coloneqq \{w:d(u,w) < r\}, B^-(u,r) \coloneqq \{w:d(w,u) < r\}$.

Given a subset of vertices $U\subseteq V$ define $E[U]$ to be the edges in $E$ with both endpoints in $U$ and $G[U] \coloneqq (U, E[U])$.

Denote by $[k]$ the set $\{1,\ldots, k\}$. For ease of notation, when referring to cycles $v_1, \ldots, v_k$ we abuse the notation of ${\text{mod }k}$ and say something is true for all pairs $(v_i, v_{i+1_{\text{mod }k}})$ to mean the pairs $(v_1, v_2) , \ldots, (v_{k-1}, v_k),(v_k, v_1)$.

Next, we define the two non-standard notions of distance we use throughout the paper. The min-distance between a pair of points in a directed graph $G$ is defined as the minimum of the two directed distances between them, $\dmin^G(u,v) \coloneqq \min(d_G(u,v), d_G(v,u))$. Again, when $G$ is clear from context we omit it. We define the min-diameter of $G$, or $\mindiam(G)$ to be the largest min-distance between a pair of points in $G$.

Lastly, we include some definitions and lemmas from prior work which we will use in our algorithms. First we define a partition of a set $C$, introduced under different notation by Dalirrooyfard et al. \cite{mindistance2019}.
\begin{definition}\label{def:partitionc}
    Let $C\subseteq V$ and $x\in C$. Assume all vertices are assigned a global, unique ID. Define $C^+_x$ to be the vertices such that their distance from $x$ is smaller than their distance to $x$, breaking ties by ID: $C^+_x \coloneqq \{v\in C : d(x,v) < d(v,x)\lor [d(x,v) = d(v,x)\land ID(x) < ID(v)]\}$. Define $C^-_x\coloneqq C \setminus (C^+_x \cup \{x\})$.
\end{definition}

The use of ID in this definition is entirely arbitrary, as when a vertex is of equal distances to $x$ in both directions we don't care which set it falls into. We only need this choice to be consistent, i.e. if $x$ falls into $C_y^+$ then $y$ falls into $C_x^-$.

The following is equivalent to Lemma 3.1 in \cite{mindistance2019}
\begin{lemma}[\cite{mindistance2019}]\label{lm:balancedpartition}
    For any $C\subseteq V$ there are $\geq \frac{|C|}{2}$ vertices $x\in C$ such that $|C^+_x|\leq \frac{8}{9}|C|$ and $|C^-_x| \leq \frac{8}{9}|C|$.
\end{lemma}

Next, we introduce a lemma used in prior work on directed girth approximation. This lemma, taken from \cite{directedgirth2020}, is attributed to \cite{chechik2020girth}.
\begin{lemma}[Lemma 8 in \cite{directedgirth2020}]\label{lm:ddirectedgirth}
    Let $G = (V,E)$ be a directed graph with $n$ nodes and integer edge weights in $[M]$. Let $S\subseteq V$ with $|S| > c \log n$ (for $c \geq 100/\log (10/9)$) and let $p> 0$. Let $R$ be a random sample of $c\log n$ nodes of $S$ and define
    \[S' \coloneqq \{a\in S : d(a,r) < p ~\forall r\in R\} = S \cap \bigcap _{r\in R}B^-(r, p).\] Suppose that for every $a\in S$ there are at most $0.2 |S|$ nodes $b\in S$ so that $d(a,b), d(b,a) < p$, then $|S'| \leq 0.8 |S|$ with probability $>1 - \frac{1}{n^{99}}$.
\end{lemma}
Note that by flipping the direction of all edges we get the same lemma for the symmetric case when $S'$ is defined by $S\cap \bigcap_{r\in R}B^+(r,p)$.

In the original statement of this lemma the condition on $S$ is that for every $a\in S$ there are at most $0.2 |S|$ nodes $b \in V$ (not $S$) such that $d(a,b),d(b,a) \leq p$ (and not $<$), their statement also required $p$ to be an integer. However, this original statement would be insufficient for our purposes. The original proof of \cite{directedgirth2020} still works under our weaker assumptions with slight adjustments; we include it here for completeness.

\begin{proof}[Proof of \cref{lm:ddirectedgirth}]
    We prove the lemma by counting the number of ordered pairs $a,b\in S$ for which $d(a,b),d(b,a) < p$. Since for every $a\in S$ the number of $b\in S$ such that $d(a,b),d(b,a) < p$ is  $\leq 0.2 |S|$, the number of such ordered pairs is $\leq 0.2|S|^2$.

    On the other hand, if $|S'|>0.8|S|$ we will show that the number of such pairs is large with high probability. Consider any $a\in S$ for which there are $\geq 0.1|S|$ nodes $b\in S$ such that $d(a,b) \geq p$, then the probability of a randomly sampled vertex $r\in S$ to have $d(a,r) < p$ is $\leq 0.9$. The probability that $a\in S'$, i.e. $d(a,r) < p$ for all $r\in R$, is hence at most $0.9^{c\log n} \leq 1/n^{100}$. Via a union bound, with probability $\geq 1- 1/n^{99}$, no such vertex is in $S'$, so every $a\in S'$ has $\geq 0.9 |S|$ vertices $b\in S$ such that $d(a,b) < p$.

    Therefore, if $|S'|>0.8|S|$, with high probability there are at least $0.8 |S|\cdot 0.9|S| = 0.72|S|^2$ ordered pairs $(a,b)\in S\times S$ with $d(a,b) < p$. There are at most ${|S|\choose 2}\leq 0.5 |S|^2$ ordered pairs $(a,b)\in S\times S$ such that exactly one of $\{d(a,b) < p, d(b,a) < p\}$ holds. So, with high probability, there are $\geq 0.22 |S|^2 > 0.2|S|^2$ ordered pairs $(a,b)\in S\times S$ with $d(a,b) < p, d(b,a) < p$, contradiction.
\end{proof}

    \section{The Type Framework}\label{sec:type}
Before constructing our min-diameter approximation algorithm we begin by creating a type-classification framework. We define two vertex types, show how to identify vertices of the same type, and prove a useful property of large sets of same-type vertices.

Given a threshold $D$, we would like to find a pair of points with $\dmin(u,v)\geq D/3$ or determine that $\mindiam(G)<D$. By binary searching over $D$ we find the best choice of $D$ and obtain a $3$-approximation to the min-diameter.

Assume we are in the case that $\mindiam(G) \geq D$; we would like to find a pair of points of min-distance at least $D/3$ in some subset $U$. Let $s,t$ be a pair of diameter endpoints, $\dmin(s,t) \geq D$ and define the \emph{type} of a vertex according to $s$ and $t$. For the sake of readability we will not index our notations by $s,t$ and $D$ but it is important to emphasize that \textit{all} the definitions in this section are with respect to a fixed (unknown) pair of diameter endpoints and a given threshold $D$.

\begin{definition}
    Given a threshold $D$ and two vertices $s,t$ such that $\dmin(s,t) \geq D$, we say that a vertex $v$ is \emph{Type 1} if $d(s,v) < D/3$ and $d(t,v) < D/3$. We denote the vertices of Type 1 by $T_1$. Similarly, we say that a vertex is \emph{Type 2} if $d(v,s) < D/3$ and $d(v,t) < D/3$ and denote the set of  vertices of Type 2 by $T_2$.
\end{definition}

Note that if a vertex $v$ is in neither $T_1$ nor $T_2$  then computing SSSP from this vertex will find a min-distance $\geq D/3$ to either $s$ or $t$, concluding our algorithm. We will call such points \emph{degenerate}.

\begin{definition}
    Given a graph $G = (V,E)$ and a subset $U\subseteq V$, we say a vertex $v$ is \emph{$U$-degenerate} in $G$ if $\max_{u\in U}\dmin^G(u,v) \geq D/3$.
\end{definition}

\begin{claim}
    Any $v\notin T_1 \cup T_2$ is $U$-degenerate for any subset such that $s,t\in U$.
\end{claim}
\begin{proof}
    Assume towards contradiction a vertex $v\notin T_1\cup T_2$ is non-degenerate. Then $\dmin(s,v) < D/3$ and so $d(s,v) < D/3$ or $d(v,s) < D/3$. Assume w.l.o.g. $d(s,v) < D/3$. Similarly, $\dmin(t,v) < D/3$ so $d(v,t) < D/3$ or $d(t,v) < D/3$. Since $v\notin T_1$ we have that $d(t,v) \geq D/3$ so we must have $d(v,t) <D/3$. But now, by the triangle inequality, $d(s,t) \leq d(s,v ) + d(v,t) < 2D/3$, contradiction.
\end{proof}

Clearly, running SSSP from a degenerate point will result in finding the desired approximation.
We therefore assume for the rest of the section that unless stated otherwise, no point from which we run SSSP is degenerate, as otherwise we would be done.

Without knowing $s$ and $t$, we cannot classify the type of vertices in the graph. However, using the following claim we can identify points that share the same type as one another.

\begin{claim}\label{clm:typeloop}
    Let $v_1, v_2, \ldots, v_k$ be non-$U$-degenerate vertices for some subset $U$ containing $s$ and $t$ and assume that $d(v_i, v_{i+1_{\text{mod }k}}) < D/3$ for $i\in [k]$. Then $v_1, \ldots, v_k \in T_1$ or $v_1, \ldots, v_k \in T_2$.
\end{claim}
\begin{proof}
    Let us first assume $v_1\in T_1$. If $v_2\in T_2$, then we have $d(s,v_1), d(v_1, v_2), d(v_2, t) < D/3$. By the triangle inequality we now have that $d(s,t) \leq d(s,v_1) + d(v_1, v_2) +d(v_2, t) < D$, in contradiction to $\dmin(s,t) \geq D$. Thus $v_2\in T_1$. In the same way, this implies that $v_3\in T_1$ and inductively we have that $v_1, \ldots, v_k \in T_1$.

    On the other hand, if $v_1 \in T_2$ then we arrive at the same contradiction if $v_k\in T_1$. Thus $v_1\in T_2$ implies $v_k\in T_2$ which implies $v_{k-1}\in T_2$ and inductively we have that $v_1,v_2, \ldots, v_k \in T_2$.
\end{proof}

In the next lemma we see the power of having a set of points of the same type, even if we don't know which of the two types it is. More specifically, if a set contains a diameter endpoint and we have the guarantee that all points in it are either degenerate or of the same type, then we can find our desired $3$-approximation in near-linear time.

\begin{lemma}\label{lm:sametypeset}
    Let $G=(V,E)$ and let $s,t\in U\subseteq V$ and $S\subseteq U$ be a set of vertices such that $\{s,t\}\cap S \neq \emptyset$ and either $S\cap T_1 = \emptyset$ or $S\cap T_2 = \emptyset$. Let $N\geq |U|$\footnote{Note that $N$ can be larger than $|V|$ here. We will need this when we run this algorithm on small subgraphs but still want to retain high probability guarantees with respect to the size $n$ of the original graph.}, $m=|E|$, then in $\Oh(m\cdot \text{\emph{polylog}}~N)$ time we can find a pair of points in $U$ of min-distance $\geq D/3$ with probability $\geq 1-1/N^3$.
\end{lemma}
We note that the above algorithm runs in $\tilde{O}(m)$ time regardless of whether the conditions are met. If the conditions are not met it is not guaranteed to succeed.

\begin{proof}[Proof of \cref{lm:sametypeset}]
We will construct the algorithm for the case when $S \cap T_2 = \emptyset$, as the other is symmetric and we can perform the algorithm for both without knowing which case we are in. Further assume w.l.o.g. that $s\in S$.

As our base case, if $|S| \leq c\log N$, for the constant $c$ from \cref{lm:ddirectedgirth}, then run SSSP from every node in $S$. This guarantees that we will find a pair of points in $U$ of min-distance $\geq D/3$.

Note that for any vertex $x\in S$ either $x$ is $U$-degenerate, in which case we can run SSSP from $x$ and find a point $y\in U$ such that $\dmin(x,y) \geq D/3$, or $s\in B^-(x, D/3)$ since $x$ is Type 1. We would therefore like to use \cref{lm:ddirectedgirth}, taking intersections of the incoming balls of sampled points until we reduce the size of $S$ enough. To do so, we need to make sure $S$ satisfies the conditions of the lemma, specifically that no $a\in S$ has
$|W(a)\coloneqq B^+(a, D/3)\cap B^-(a, D/3)\cap S|\geq 0.2|S|$.

Let $Z$ be a random sample of size $ 5\log_{10/9} N$ vertices in $S$. For any $a\in S$ such that $|W(a)|\geq 0.1|S|$, each vertex in $S$ hits $W(a)$ with probability $\geq 0.1$, so with probability $\geq 1 - (0.9)^{5\log_{10/9} N} = 1 - 1/N^5$ some vertex in $Z$ hits $W(a)$. With probability $\geq 1- 1/N^4$, $Z$ hits $W(a)$ for all $a\in S$ with $|W(a)|\geq 0.1|S|$.

Run SSSP from every vertex in $Z$. If any vertex $z\in Z$ finds a vertex $u\in U$ such that $\dmin(z,u)\geq D/3$, return the pair $(u,z)$. Otherwise, let $\hat{S}$ be the set $S$ after removing $\cup_{z\in Z}W(z)$. If $|\hat{S}|<\frac{|S|}{2}$, skip the rest of this iteration of the algorithm as we have already shrunk $S$ sufficiently.

Now, if there was a vertex $z\in Z$ such that $s\in W(z)$, then $\dmin(z,t) \geq 2D/3$ and we would have found a sufficiently far apart pair of vertices in $U$. Thus, with high  probability (in $N$), if we did not find such a pair, after this step $s\in \hat{S}$ and every $a\in \hat{S}$ has $|W(a)|<0.1 |S|\leq 0.2|\hat{S}|$. Take $S$ to be $\hat{S}$.

We can now randomly sample a set $R$ of $c\log N$ vertices from $S$ and define $S'$ as in \cref{lm:ddirectedgirth} to be the intersection of their incoming balls of radius $D/3$, $S' \coloneqq S \cap \bigcap_{r\in R} B^-(r,D/3) $. If any vertex in $R$ is $U$-degenerate, we find a pair of vertices in $U$ with large min-distance and return them. Otherwise, all vertices in $R$ are of Type 1 and so $s\in S'$. By \cref{lm:ddirectedgirth}, $|S'|\leq 0.8|S|$.

If we have not finished, we have a set $S'$ that satisfies the same conditions as $S$ and is smaller by a factor of $0.8$. We can therefore repeat this algorithm (first ensuring all $W(a)$ are small then sampling a set $R$) $\log_{5/4} N$ times until $|S|<c\log N$, at which point we compute SSSP from all of its vertices.

If $s\in S$, and the set contained no points of Type 2, then every time we constructed the set $S'$ we retained $s$, so the algorithm will find at some point a pair of vertices in $U$ of min-distance $\geq D/3$.

We repeat the algorithm again, this time assuming no vertices are Type 1 and taking the intersection of $B^+(r, D/3)$.  The algorithm runs in time $\Oh(m\cdot  \text{polylog}N)$ and by a union bound succeeds with probability $\geq 1 - 1/N^3$.

\end{proof}

Combining \cref{clm:typeloop} and \cref{lm:sametypeset} gives us the following algorithm, which will be a useful subroutine in our main min-diameter approximation algorithm.

\begin{lemma}\label{lm:bigoverlap}
    Given a graph $G=(V,E)$, let $U\subseteq V$ be a subset of vertices containing $s,t$ and $N\geq |U|$. Let $x\in U$ and let $A\subseteq U$ be a subset such that for every $a\in A$ we have $\ell$ vertices\footnote{The set $Q(a)$ can also be smaller than $\ell$ by repeating vertices.} $Q(a) = \{a = u_1, u_2, \ldots, u_\ell\} \subseteq U$ such that
    \begin{enumerate}
        \item $x\in Q(a)$,
        \item $d(u_i, u_{i+1_{\text{mod }\ell}}) < D/3$ for $i\in [\ell]$.
    \end{enumerate}
    Then if $\{s,t\}\cap A \neq \emptyset$, in $\Oh(\ell \cdot m\cdot \text{\emph{polylog}}~N)$ time we can find a pair of points in $U$ with min-distance $\geq D/3$ with probability $\geq 1-1/N^3$. Denote this algorithm by $\singletypealg (G, U,A,x, \ell, D, N)$.
\end{lemma}

We again note that the algorithm in the above lemma runs in $\Oh(\ell \cdot m\cdot \text{polylog}~N)$ time even when the conditions are not met.

\begin{proof}[Proof of \cref{lm:bigoverlap}]
    We perform the same algorithm as in \cref{lm:sametypeset}, taking $S = A$, with one change. Every time we sample a point $a\in A$ and run SSSP from it, we additionally run SSSP from every point in $Q(a)$. If any of these points is $U$-degenerate, we have found a pair of points in $U$ with min-distance $\geq D/3$ and we return it. Otherwise, by \cref{clm:typeloop}, every sampled point $a\in A$ has the same type as $x$. As with \cref{lm:sametypeset}, we run the algorithm twice, first assuming all points are Type 1 (or degenerate) and then assuming all points are Type 2. Thus, if $x\in T_1$ the algorithm succeeds on the first run and otherwise it succeeds on the second run.

    The algorithm succeeds with probability $\geq 1 - 1/N^3$ as before. The runtime incurs a blowup of $\ell$, as now every call to SSSP is replaced by $\ell$ such calls.
\end{proof}

    \section{3-Approximation to Min-Diameter}\label{sec:mindiam}

We can now construct our main algorithm and prove \cref{thm:mindiam}. Let $D$ be a threshold and assume again that $\mindiam(G)\geq D$, $\dmin(s,t) \geq D$ and we would like to find a pair of points with min-distance $\geq D/3$.

We will start with the set $C=V$ as our set of potential diameter endpoints. At each step we will remove some points from $C$ and then divide it into two balanced sets $C^+, C^-$ with the guarantee that if $s,t$ were in $C$ at the beginning of the algorithm then either $s,t\in C^+$ or $s,t\in C^-$. We will then recurse on the two smaller sets.

The algorithm receives as input the depth of the recursion $i$, a diameter estimate $D$, the size of the original graph $N$ and a value $L$ that depends on the size of the original graph and remains constant as the graph gets smaller, a set of potential diameter endpoints $C$ and up to $L+1$ sets of padding. At depth $i$ of the recursion we will have $L - i+1$ padding sets, with the last $i$ sets being empty.

Given an $n$-node $m$-edge graph $G = (V,E)$, we set $N=n, L = \ceil{\log_{9/8} n}$ and first call the algorithm on the input: $\mindiamalg(i=0, D,N, L,C = V, P_1 = V, P_2 = V, \ldots, P_{L+1} = V) $. For the full pseudocode see \cref{alg:mindiam3approx}.

\begin{algorithm}
     \caption{$\mindiamalg(i, D, N, L, C, P_1, \ldots, P_{L+1})$ }\label{alg:mindiam3approx}
     \begin{algorithmic}[1]
        \State All distances in this algorithm are computed in $G[P_{L-i+1}]$: every construction of a ball around a point or set involves an implicit call to SSSP on $G[P_{L-i+1}]$. If this call finds a pair of points in $P_{L-i}$ of min-distance $\geq D/3$, return them. \label{ln:startmindiamalg}
        \State If $|C| \leq 2$, compute the distance between the points $y,z\in C$. Return the pair if $\dmin(y,z) \geq D/3$ and null otherwise.
        \State Sample $x\leftarrow C$.\label{ln:samplingx}
        \State Compute $C^+_x, C^-_x$, if either has size $> \frac{8}{9}|C|$, resample $x$ and try again, after $3\log N$ failed attempts return null. \label{ln:resamplex}
        \State $P_0^+ \leftarrow C^+_x, P_0^- \leftarrow C^-_x$
        \For {$j = 1,\ldots, L-i$}
            \State $P^+_j \leftarrow B^+(P^+_{j-1}, D/3)\cap P_j$
            \State $P^-_j \leftarrow B^-(P^-_{j-1}, D/3)\cap P_j$
        \EndFor
        \State $W_x \leftarrow P_{L-i}^+ \cap P_{L-i}^-$
        \State For every $w\in W_x$ store
        $2(L-i-1)$ vertices $w_1^+, \ldots, w_{L-i-1}^+,w_1^-, \ldots, w_{L-i-1}^-\in P_{L-i-1}$ such that $d(x, w_1^+), d(w_1^+,w_2^+), \ldots, d(w_{L-i-1}^+, w), d(w, w_{L-i-1}^-), \ldots, d(w_2^-, w_1^-), d(w_1^-, x) < D/3$.
        \If {$|E[W_x]| > |E[P_{L-i+1}]| / L$}
            \State $A_{L-i}^+ \leftarrow W_x, A_{L-i}^-\leftarrow W_x $
            \For{$j = L-i-1, L-i-2, \ldots, 0$}
                \State $A_j^+ \leftarrow P_j \cap B^-(A_{j+1}^+, D/3)$
                \State $\forall a^+_j\in A_j^+$ store a vertex $a_{j+1}^+\in A_{j+1}^+$ such that $d(a_j^+, a_{j+1}^+) < D/3$.
                \State $A_j^- \leftarrow P_j \cap B^+(A_{j+1}^-, D/3)$
                \State $\forall a_j^-\in A_j^-$ store  a vertex $a_{j+1}^-\in A_{j+1}^-$ such that $d(a_{j+1}^-, a_{j}^-) < D/3$.
            \EndFor
            \State $A\leftarrow A_0^+ \cup A_0^-$
            \State $(y,z)\leftarrow \singletypealg (G[P_{L-i+1}], P_{L-i},A,x, \Oh(L), D, N)$ \label{ln:callsingletypealg}
            \State \Return $(y,z)$ if not null. \State Otherwise, $C\leftarrow C\setminus A$, go to line \ref{ln:samplingx}.\label{ln:removeA}
        \EndIf \label{ln:endbigoverlap}
        \State $(y_1, z_1) \leftarrow \mindiamalg(i+1, D,N, L, C_x^+, P_1^+, \ldots, P_{L-i}^+, \emptyset, \ldots, \emptyset)$
        \State $(y_2, z_2) \leftarrow \mindiamalg(i+1, D, N, L, C_x^-, P_1^-, \ldots, P_{L-i}^-, \emptyset, \ldots, \emptyset)$
        \State \Return $(y_1, z_1)$ or $(y_2, z_2)$ which isn't null.
        \end{algorithmic}
\end{algorithm}

At depth $i$ we run all of our distance computation on $G[P_{L-i+1}]$. For ease of notation, denote by $m$ the number of edges in this graph, $m = |E[P_{L-i+1}]|$.

We begin by sampling a point $x\in C$ and compute $C_x^+, C_x^-$ according to \cref{def:partitionc}. By \cref{lm:balancedpartition}, with probability $\geq \frac{1}{2}$ we have an $x$ such that $|C_x^+|,|C_x^-| \leq \frac{8}{9}|C|$. If this is not the case, we resample. After $3\log N$ failed attempts we give up. This step succeeds with probability $\geq 1-\frac{1}{N^3}$.

Next, we define $L-i$ layers of padding. We take $P_1^+$ to be the outgoing $D/3$ neighborhood of $C^+_x$ in $P_1$.
Similarly, we take $P_1^-$ to be the incoming $D/3$ neighborhood of $C_x^-$ in $P_1$.
We repeat this $L-i$ times: taking $P^+_j \coloneqq B^+(P^+_{j-1}, D/3)\cap P_{j}, P^-_j \coloneqq B^-(P^-_{j-1}, D/3)\cap P_{j}$ for $2\leq j\leq L-i$.

Now we would like to recurse. If $W_x\coloneqq P_{L-i}^+\cap P_{L-i}^-$ contains $\leq \frac{m}{L}$ edges, we call $\mindiamalg$ twice. Once on $\mindiamalg(i+1, D,N,L, C_x^+, P_1^+, \ldots, P_{L-i}^+, \emptyset, \ldots, \emptyset)$
and another time on $\mindiamalg(i+1, D,N,L, C_x^-, P_1^-, \ldots, P_{L-i}^-, \emptyset, \ldots, \emptyset)$.

If $s,t\in C$ then either $s,t\in C_x^+$, in which case we will show that the first call finds a pair of points of min-distance $\geq D/3$, or $s,t\in C_x^-$, in which case the second call finds such a pair.

We are left to handle the case where $|E[W_x]| > \frac{m}{L}$. We note that by the invariants that our algorithm maintains, any pair of points in $ P_{L-i}$ (including $C, W_x$ and padding sets in between them) of min-distance $<D/3$ in $G$ will have min-distance $<D/3$ in $G[P_{L-i+1}]$ so if we find a pair of points in this set of min-distance $\geq D/3$ we return it and complete the algorithm.

To handle the large overlap, we isolate a set of vertices in $C$ that are contained in a loop (in $P_{L-i}$) between $x$ and $W_x$. If $s,t\in C$, then by \cref{clm:typeloop}, all these points will be either $C$-degenerate in $G[P_{L-i+1}]$ (and thus we will have found a pair of far away points in $P_{L-i}$) or of the same type. We take our set $A\subseteq C$ to include points in $C^+_x$ that can reach $W_x$ within distance $<(L-i)D/3$  and points in $C_x^-$ that can be reached from $W_x$ within distance $<(L-i)D/3$.

More precisely\footnote{This is almost equivalent to taking $A = C^+_x\cap B^-(W_x, (L-i)D/3) \cup C_x^-\cap B^+(W_x, (L-i)D/3)$, but allows us to find the points along the path more directly without the potential of introducing additive errors.}, take $A_{L-i-1}^+ \coloneqq P_{L-i-1}\cap B^-(W_x, D/3)$ and $A_{L-i-1}^- \coloneqq P_{L-i-1}\cap B^+(W_x, D/3)$. For $j = L-i-2, \ldots, 0$ define $A_{j}^+ \coloneqq P_{j}\cap B^-(A_{j+1}^+, D/3)$ and $A_j^- \coloneqq P_{j}\cap B^+(A_{j+1}^-, D/3)$. Finally, take $A = A_0^+ \cup A_0^- \subseteq C$.

For each vertex $a\in A$ we can store a set of $k=\Oh(L)$ points $a=a_0, \ldots, a_{k-1} \in P_{L-i}$ such that $x\in \{a_0, \ldots, a_{k-1}\}$ and $d_{G[P_{L-i+1}]}(a_j, a_{j+1_{\text{mod} k}}) < D/3$ for $j = 0, \ldots, k-1$ as follows. If $a\in A_0^+$ we take a point $a_1\in A_1^+$ that $a$ can reach in distance $<D/3$, followed by points $a_2\in A_2^+,\ldots, a_{L-i-1}\in A_{L-i-1}^+, w\in W_x$ that can each reach the next in distance $<D/3$. We can then take $L-i$ points in the $L-i$ layers of padding that allow $w$ to reach $x$ in steps consisting of paths of length $<D/3$, ending with $x$. Similarly, we can find such a set for any point in $A\cap A_0^-$.

We now have all we need to be able to run the algorithm $\singletypealg$ from \cref{lm:bigoverlap}, on $P_{L-i}$, the set $A$ and point $x$, with $k$ points for each vertex in $A$, in the graph $G[P_{L-i+1}]$. It runs in time $\Oh(m L\cdot \text{polylog}~N)$ and if $s\in A$ or $t\in A$ it finds a pair of points in $P_{L-i}$ with min-distance $\geq D/3$ in $G[P_{L-i+1}]$ with probability $\geq 1-\frac{1}{N^3}$.

Thus, if $\singletypealg(G[P_{L-i+1}],P_{L-i}, A, x, k, D, N)$ didn't return a pair of far apart vertices, we can assume that $s,t\notin A$ and remove $A$ from $C$ (without removing the vertices from $P_1,\ldots, P_{L-i+1}$, as they can still be used in shortest paths).

We now restart the algorithm with the smaller set $C'\subseteq C$ - sample a new $x'\in C'$, resample until the sets $C_{x'}^+, C_{x'}^-$ are not too large and compute $W_{x'}$. We note that $W_x,W_{x'}$ are disjoint. Otherwise, if $x'\in C_x^+$ then it would be able to reach $W_x$ and would have been removed from $C$. Similarly, if $x'\in C_x^-$ it would be reached from $W_x$ and thus removed from $C$.

Therefore, the sets $W_x$ are disjoint over the different iterations of the algorithm. In particular, the edges contained within each set are disjoint. Since each $|E[W_x]|\geq m/L$, after $\leq L$ iterations we will find a vertex $x$ with $|E[W_x]|<m/L$ and proceed with the recursion.

Thus, our algorithm runs in time $\Ohtilde(L\cdot m) = \Ohtilde(m)$ before calling the recursive step. At each level of the recursion the number of edges blows up by a factor of at most $(1 + 1/L)$. Since $|C|$ shrinks by a factor of $8/9$ every call, at depth $L$ we have $|C|\leq (8/9)^L\cdot n \leq 2$, so we do not reach any recursion depth greater than $L$. Thus, the total number of edges at each stage of the recursion is bounded by $(1+1/L)^L\cdot |E| = \Oh(|E|)$, giving a total runtime of $\Ohtilde(|E|)$.
In the following claims we formally argue correctness and runtime.

\paragraph{Correctness.} We argue for the correctness of the algorithm using the following two claims. First we show that at every depth of the recursion at least one call will retain a set $C$ with a pair of points of min-distance $\geq D$. Next, we argue that any pair of points with distance $<D/3$ will also have distance $<D/3$ in the recursive calls. These claims guarantee that if $\mindiam(G)\geq D$ then we will eventually find a pair of points of min-distance $\geq D/3$.
Since at recursion depth $i$ all the points we consider are in $P_{L-i}$ and we compute distances in $G[P_{L-i+1}]$, any pair of points we return will in fact have the desired distance in $G$.

\begin{claim}\label{clm:cpartitionsameside}
    If there exist a pair of points $s,t\in C$ such that $\dmin^G(s,t) \geq D$ then during the recursive calls either $s,t\in C_x^+$ or $s,t\in C_x^-$.
\end{claim}

\begin{proof}
    Denote by $P$ the set $P_{L-i+1}$ on which we run our graph searches. Since $\dmin^G(s,t)\geq D$ we also have that $\dmin^{G[P]}(s,t) \geq D$.
    Now note that $C_x^+, C_x^-$ partition the set $C$. If $s,t$ are not in the same set then assume w.l.o.g. that $s\in C_x^-, t\in C_x^+$. If $d_{G[P]}(s,x) \geq D/3$ or $d_{G[P]}(x,t) \geq D/3$ then we have found a pair of points in $P_{L-i}$ with min-distance $\geq D/3$ in line \ref{ln:resamplex} and returned them. Otherwise, we have that $d_{G[P]}(s,t) \leq d_{G[P]}(s,x) + d_{G[P]}(x,t) < 2D/3$, contradiction.

    By \cref{lm:sametypeset}, if $s\in A$ or $t\in A$ when we call $\singletypealg$, then we will complete the algorithm. Thus, when we remove vertices from $C$ in line \ref{ln:removeA} we retain that $s,t\in C$ at the beginning of the next iteration.
\end{proof}

\begin{claim}\label{clm:pathscontained}
    At recursion depth $i$, $C = P_0 \subseteq P_1 \subseteq \ldots \subseteq P_{L-i+1}$ and for any $j\in [L-i]$, any path of length $<D/3$ between a pair of points $x,y\in P_{j}$ is contained in $P_{j+1}$.
\end{claim}

\begin{proof}
    We will prove this inductively. At the beginning of iteration $i=0$ all sets are equal to $V$ so the claim holds. Now assume it holds at the beginning of an iteration with $i-1$. First, for any $x$, $C^+_x \subseteq P_1^+ \subseteq \ldots \subseteq P_{L-i}^+$ by definition. Similarly,  $C_x^- \subseteq P_1^-\subseteq \ldots \subseteq P_{L-i}^-$.

    For any $j\in [L-i]$ and any pair of points $x,y\in P_j^+$, any path of length $<D/3$ between them is contained in $P_{j+1}\subseteq P_{L-i+1}$ by our assumption. Thus, any such path is contained in $B^+_{G[P_{L-i+1}]}(P_j^+, D/3)$ and so is contained in $P_{j+1}^+$. Similarly, any path of length $<D/3$ between a pair of points in $P_j^-$ is contained in $P_{j+1}^-$. Thus these conditions are maintained on the next call to the algorithm.
\end{proof}

\begin{cor}
    For any pair of points $x,y\in P_{L-i-1}^+$  we have \[\dmin^{G[P_{L-i+1}]}(x,y)< D/3 \Rightarrow \dmin^{G[P^+_{L-i}]}(x,y) < D/3.\] Likewise, for any pair of points $x,y\in P_{L-i-1}^-$ we have \[\dmin^{G[P_{L-i+1}]}(x,y)< D/3 \Rightarrow \dmin^{G[P^-_{L-i}]}(x,y) < D/3.\]
\end{cor}

\begin{proof}
    We will prove the first statement, as the second is symmetric. If $\dmin^{G[P_{L-i+1}]}(x,y)< D/3$, then there is a path of length $<D/3$ between $x,y$ in $P_{L-i+1}$. Thus, all vertices in this path will be contained in the ball of radius $D/3$ around $x,y$, which is in turn contained in  $P_{L-i}^+ = B^+_{G[P_{L-i+1}]}(P_{L-i-1}^+, D/3)\cap P_{L-i}$. Therefore, $\dmin^{G[P^+_{L-i}]}(x,y) < D/3$.
\end{proof}

\paragraph{Runtime.} Finally, we argue that our algorithm runs in $\Ohtilde(|E|)$ time. First we show that the sets $W_x$ are in fact disjoint.

\begin{claim}
    Let $x_1, x_2$ be two points sampled in line \ref{ln:samplingx}. Then $W_{x_1}\cap W_{x_2}=\emptyset$.
\end{claim}

\begin{proof}
    Assume w.l.o.g. that $x_1$ was sampled at an earlier iteration than $x_2$ and assume towards contradiction that $\exists w\in W_{x_1}\cap W_{x_2}$. By definition of $W_{x_2}$, there exist $a_0^+\in C,a_1^+\in P_1, \ldots, a_{L-i-1}^+\in P_{L-i-1}, a_{L-i-1}^-\in P_{L-i-1}, \ldots a_1^-\in P_1, a_0^-\in C$ such that
    \begin{align*}
    d(x_2, a_0^+), d(a_0^+, a_1^+), \ldots, d(a_{L-i-2}^+, a_{L-i-1}^+), d(a_{L-i-1}^+, w)&, \\d(w, a_{L-i-1}^-), d(a_{L-i-1}^-, a_{L-i-2}^-), \ldots, d(a_1^-, a_0^-),d(a_0^-, x_2)& < D/3.
    \end{align*}

    Now consider the way that $C$ was divided after sampling $x_1$. If $x_2\in C_{x_1}^+$ then $a_0^+, \ldots, a_{L-i-1}^+, w\in W_{x_1}$ are evidence that $x_2\in A_0^+$, since the sets $P_j$ don't change within one call to $\mindiamalg$. Otherwise, $x_2\in C_{x_1}^-$, then $a_0^-, \ldots, a_{L-i-1}^-$ are evidence that $x_2\in A_0^-$. Thus, we would have $x_2\in A$ and $x_2$ would have been removed from $C$ at this point and not sampled at a later time, contradiction.
\end{proof}

Now we can bound the number of iterations of lines \ref{ln:samplingx}-\ref{ln:endbigoverlap} by $L$. The runtime of each iteration is dominated by the runtime of line \ref{ln:callsingletypealg}, $\Oh(mL\cdot \text{polylog}N)$. Thus, lines $\ref{ln:startmindiamalg}$-\ref{ln:endbigoverlap} run in time $\Oh(L^2m\cdot \text{polylog}N)$, where recall that $m$ denotes the number of edges in this instance, $|E[P_{L-i+1}]|$.

Next, we bound the total number of edges across all iterations of the recursion at depth $i$.

\begin{claim}
    The total number of edges across all instances of the recursion at depth $i$ is at most $(1+1/L)^i\cdot |E|$.
\end{claim}

\begin{proof}
    We prove this by induction. At iteration $i=0$ we have $|E|$ edges and one instance. Now assume the claim is true for $i-1$: we have $r$ instances of the algorithm running, with $m_1, m_2, \ldots, m_r$ edges in each instance such that $\sum_{j=1}^r m_j \leq (1+1/L)^{i-1}|E|$. Each instance $j$ is split into two instances such that the number of edges that overlap these two instances is $\leq m_j/L$. Thus, the total number of edges in these two instances combined is $\leq (1+1/L)m_j$. Therefore, the total number of edges across all instances at recursion depth $i$ is bounded by
    \[
    \sum_{j=1}^r (1+1/L)m_j \leq (1+1/L)^i |E|.
    \]
\end{proof}

Now, since the runtime of each instance before its recursive call is $\Oh(L^2m \cdot \text{polylog}N)$, we can bound the total runtime of all instances of the recursion at depth $i$ by \[\Oh((1+1/L)^i|E|\cdot L^2\cdot \text{polylog}N).\] Next, we bound the recursion depth by $L$.

\begin{claim}
    If $i \geq L$ then $|C| \leq 2$.
\end{claim}
\begin{proof}
    At every iteration, we recurse on the sets $C_x^+,C_x^-$ of size $\leq \frac{8}{9}|C|$. When $i=0$ we have $|C|=n$ and so by induction at iteration $i$ we have $|C|\leq \left(\frac{8}{9}\right)^i\cdot n$. Therefore, by our choice of $L \geq \log_{9/8} n$, if $i=L$ we have $|C|\leq 2$.
\end{proof}

Therefore, we have $\leq L$ levels of recursion. At each level we have a total runtime of $\Ohtilde((1+1/L)^i|E|\cdot L^2)$ so our final runtime is bounded by
\[
    \Oh((1+1/L)^L|E|\cdot L^3\cdot \text{polylog}N) = \Ohtilde(|E|).
\]

\paragraph{Success Probability.} Two steps of our algorithm are randomized - line \ref{ln:resamplex} and our call to $\singletypealg$ in line $\ref{ln:callsingletypealg}$. Since each succeeds with probability $\geq 1 - \frac{1}{N^3}$, for the original size of the graph $N=n$, and we run the algorithm $\tilde{O} (nL)$ times, using a union bound our algorithm succeeds with probability $\geq 1 - \frac{1}{n}$.

    \bibliographystyle{alphaurl}
    \bibliography{refs}
\end{document}